\documentclass[conference,letterpaper]{IEEEtran}

\usepackage[utf8]{inputenc} 
\usepackage[T1]{fontenc}
\usepackage{url}
\usepackage{ifthen}
\usepackage{cite}
\usepackage[cmex10]{amsmath} 

\usepackage{bm}
\usepackage{amsfonts}
\usepackage{amssymb}
\usepackage{mathtools}
\mathtoolsset{showonlyrefs=true}
\usepackage{cite}
\usepackage{multirow}
\usepackage{color}
\newtheorem{definition}{Definition}
\newtheorem{theorem}{Theorem}
\newtheorem{assumption}{Assumption}
\newtheorem{proof}{Proof}
\newtheorem{corollary}{Corollary}

\newtheorem{example}{Example}
\newtheorem{remark}{Remark}

\newcommand{\qed}{\hfill$\Box$}

\begin{document}
\title{A Unified Representation of Density-Power-Based Divergences Reducible to M-Estimation} 

\author{%
  \IEEEauthorblockN{Masahiro Kobayashi}
  \IEEEauthorblockA{Information and Media Center\\
                    Toyohashi University of Technology\\
                    Toyohashi, Aichi, Japan\\
                    Email: kobayashi@imc.tut.ac.jp}
  }

\maketitle

\begin{abstract}
Density-power-based divergences are known to provide robust inference procedures against outliers, and their extensions have been widely studied.
A characteristic of successful divergences is that the estimation problem can be reduced to M-estimation.
In this paper, we define a norm-based Bregman density power divergence (NB-DPD)---density-power-based divergence with functional flexibility within the framework of Bregman divergences that can be reduced to M-estimation.
We show that, by specifying the function $\phi_\gamma$, NB-DPD reduces to well-known divergences, such as the density power divergence and the $\gamma$-divergence.
Furthermore, by examining the combinations of functions $\phi_\gamma$ corresponding to existing divergences, we show that a new divergence connecting these existing divergences can be derived.
Finally, we show that the redescending property, one of the key indicators of robustness, holds only for the $\gamma$-divergence.
\end{abstract}

\section{Introduction}
Divergence functions are measures for quantitatively evaluating the difference between two probability distributions, and they play a key role in fields such as statistics, machine learning, and information theory.
Density power divergence (DPD) \cite{beta-div} (also called $\beta$-divergence \cite{Cichocki2010}), which enables substitution by the empirical distribution and provides a robust inference procedure against outliers, is widely used as a standard method in robust estimation.
This divergence has a nonnegative parameter $\gamma\geq0$ and smoothly connects the Kullback--Leibler (KL)-divergence ($\gamma=0$), which corresponds to the maximum likelihood estimation with the $L_2$-divergence \cite{scott2001} ($\gamma=1$), which corresponds to the robust estimation.
Consequently, adjusting the parameter $\gamma$ enables control over the trade-off between robustness to outliers and efficiency under the model.

Another well-known density-power-based divergence is the $\gamma$-divergence \cite{gamma-div, type0-div}.
Like DPD, $\gamma$-divergence enables substitution by the empirical distribution and the parameter $\gamma$ controls robustness against outliers and efficiency under the model.
Moreover, it can reduce latent bias close to zero even when the proportion of outliers in the data is large \cite{gamma-div, fujisawa2013}, i.e., it has higher robustness than the DPD.

Following the success of these density-power-based divergences, several extensions have been investigated.
Jones et al. \cite{type0-div} showed that by replacing the logarithmic function in the $\gamma$-divergence with a power function, both the DPD and $\gamma$-divergence can be unified.
Ray et al. \cite{FDPD} defined the functional density power divergence (FDPD) by replacing the logarithmic function with a more general function and investigated the valid conditions for it to qualify as a divergence.
Kuchibhotla et al. \cite{bridge-div} proposed the bridge density power divergence (BDPD) by taking a convex combination of the estimating equations for the DPD and $\gamma$-divergence.
Gayen et al. \cite{Gayen2024} extended this approach by considering linear combinations with nonnegative weights, thereby broadening the parameter range of the BDPD.
Focusing on the fact that both the DPD and the $\gamma$-divergence enable empirical distribution substitution and remain invariant under affine transformations of random variables, Kanamori and Fujisawa \cite{holder-div1} proposed the H\"{o}lder divergence (HD). Furthermore, they derived the Bregman--H\"{o}lder divergence (BHD) as the intersection of the HD and Bregman divergence \cite{holder-div1}.

The above the general divergences all belong to the class of non-kernel divergences \cite{Jana2019}, which enable the substitution of the empirical distribution.
However, only a more restricted subset of these non-kernel divergences---those that can be reduced to M-estimation---has been actively studied for estimation purposes.
M-estimation \cite{robust, huberbook} is a method for estimating a model $p_\theta(x)$ by solving the following estimating equation for the parameter $\theta$:
\begin{align}
    \frac{1}{n}\sum_{i=1}^n \psi(x_i, \theta)=0. \label{eq:m_estim_ee}
\end{align}
The estimating equation is obtained by taking the partial derivative of the loss function with respect to $\theta$ and setting it to zero.
When the estimation problem can be expressed as M-estimation, the estimator’s properties follow the established asymptotic theory of M-estimation \cite[Ch. 5]{asympt}.

In this paper, we construct a density-power-based divergence that has the functional flexibility of the Bregman divergence using the $L_{1+\gamma}$-norm and strictly increasing and convex function.
We call this divergence {\it norm-based Bregman density power divergence (NB-DPD)}.
NB-DPD has two degrees of freedom: the power parameter $\gamma\geq0$ and strictly increasing and convex function $\phi_\gamma$, where corresponding estimation method is reduced to M-estimation.
We show that NB-DPD converges to the KL-divergence in the limit as $\gamma \to 0$, and by specifying the strictly increasing and convex function $\phi_\gamma$, it reduces to the DPD, $\gamma$-divergence, and their generalizations. 
We also discuss its relationship with other divergences with functional flexibility. 
Finally, we discuss the robustness of the corresponding estimators in terms of their influence functions \cite{robust, huberbook}, a standard measure of robustness against outliers.
\section{Related divergences}
\subsection{Notation and Definitions}
In this paper, $\langle \cdot \rangle$ denotes an integral with respect to the Lebesgue measure.
Specifically, $\langle p \rangle = \int p(x)\, d\nu(x)$.
Moreover, $\|p\|_{1+\gamma} = \langle p^{1+\gamma} \rangle^{\frac{1}{1+\gamma}}$ represents the $L_{1+\gamma}$-norm, where the power parameter $\gamma$ is a nonnegative real number.
The set of probability distribution is denoted as $\mathcal{P}$.
The set of nonnegative real number is denoted as $\mathbb{R}_+$.
\begin{definition}[Cross-entropy and divergence]
    For any two probability distributions $p,q \in \mathcal{P}$, a functional 
    $d : \mathcal{P} \times \mathcal{P} \to \mathbb{R}$ is called a cross-entropy if it satisfies the following:
    \begin{align*}
        d(q,p) &\geq d(q,q),\\
        d(q,p) &= d(q,q) \ \Leftrightarrow\ q=p \ (a.s.).
    \end{align*}
    We define a divergence by taking the difference of the cross-entropy and ensuring that the following properties are satisfied:
    \begin{align*}
        D(q,p) &= d(q,p) - d(q,q) \geq 0,\\
        D(q,p) &= 0 \ \Leftrightarrow\ q=p \ (a.s.).
    \end{align*}
\end{definition}
\begin{definition}[Equivalence of cross-entropies]\label{def:eqv}
    Two cross-entropies $d(q,p)$ and $\Tilde{d}(q,p)$ are said to be equivalent if there exists a strictly increasing function $\xi : \mathbb{R} \to \mathbb{R}$ such that 
    $\Tilde{d}(q,p) = \xi\bigl(d(q,p)\bigr)$ for all $p,q \in \mathcal{P}$.
\end{definition}
Note that Definition \ref{def:eqv} is the same as \cite[Def. 2.2]{holder-div1}.
Cross-entropies belonging to the same equivalence class yield the same minimizer.
In other words, $\min_p d(q,p)$ and $\min_p \Tilde{d}(q,p)$ have the same solution.

\subsection{Bregman divergence and special cases}
For any $t \in (0,1)$ and $p,q \in \mathcal{P}$ with $p \neq q$, if the following inequality holds,
we say that $\Phi$ is a strictly convex functional:
\begin{align*}
    \Phi(tq(x) + (1-t)p(x)) < t\,\Phi(q(x)) + (1-t)\,\Phi(p(x)).
\end{align*}

\begin{definition}[Bregman cross-entropy and divergence \cite{Evgeni2018, Frigyik2008}]\label{def:bregman}
    Using a strictly convex functional $\Phi: \mathcal{P} \to \mathbb{R}$,
    the Bregman cross-entropy between two probability distributions is defined as follows:
    \begin{align}
        d_\Phi(q,p) = -\Phi(p) - \int\Phi_p^*(x)\left(q(x) - p(x)\right)\,d\nu(x). \label{eq:bregman_cross_ent}
    \end{align}
    Moreover, from $D_\Phi(q,p)=d_\Phi(q,p)-d_\Phi(q,q)$, we define the Bregman divergence by
    \begin{align}
        D_\Phi(q,p) = \Phi(q) - \Phi(p) \;-\; \int \Phi_p^*(x)\,\bigl(q(x) - p(x)\bigr)\,d\nu(x),
        \label{eq:bregman_div}
    \end{align}
    where $\Phi_p^*: S \to \mathbb{R},\, S \subseteq \mathbb{R}^d$ is a subgradient of the strictly convex functional $\Phi$ at $p$.
\end{definition}
The Bregman divergence can be reduced to M-estimation \eqref{eq:m_estim_ee} \cite{Gneiting2007}.
When the strictly convex functional $\Phi$ is defined by a strictly convex function $\varphi: \mathbb{R}_+ \to \mathbb{R}$, i.e., $\Phi(q) = \langle \varphi(q)\rangle$, the Bregman divergence \eqref{eq:bregman_div} is referred to as the separable Bregman divergence \cite{holder-div1}.
Most existing studies have focused on separable Bregman divergences, exploring the choice of $\varphi$ that provides a better trade-off between efficiency and robustness (e.g., \cite{beta-div, bexp-div, Roy2019, Singh2021, u-boost}).
The DPD is an example of a separable Bregman divergence.
The density power cross-entropy (DPCE) and DPD are respectively obtained from the Bregman cross-entropy \eqref{eq:bregman_cross_ent} and divergence \eqref{eq:bregman_div} by setting $\Phi(q) = \bigl(\langle q^{1+\gamma}\rangle-1\bigr)/\gamma$:
\begin{align}
    &d_\gamma^{\rm DP}(q,p) = -\frac{1+\gamma}{\gamma}\left\langle qp^\gamma \right\rangle + \frac{1}{\gamma} + 
    \left\langle p^{1+\gamma}\right\rangle, \label{eq:dpce}\\
    &D_\gamma^{\rm DP}(q,p) = \frac{1}{\gamma}\left\langle q^{1+\gamma}\right\rangle \;-\; \frac{1+\gamma}{\gamma}\left\langle qp^\gamma \right\rangle \;+\; \left\langle p^{1+\gamma}\right\rangle. \label{eq:dp}
\end{align}
Similarly, the pseudo-spherical cross-entropy (PSCE) and pseudo-spherical divergence (PSD) are obtained from the Bregman cross-entropy \eqref{eq:bregman_cross_ent} and divergence \eqref{eq:bregman_div} by setting $\Phi(q)=\bigl(\|q\|_{1+\gamma}-1\bigr)/\gamma$:
\begin{align}
    &d_\gamma^{\rm PS}(q,p) = -\frac{1}{\gamma}\,\frac{\left\langle qp^{\gamma}\right\rangle}{\| p\|_{1+\gamma}^\gamma} + \frac{1}{\gamma}, \label{eq:psce}\\
    &D_\gamma^{\rm PS}(q,p) = \frac{1}{\gamma}\,\|q\|_{1+\gamma} \;-\; \frac{1}{\gamma}\,\frac{\left\langle qp^\gamma \right\rangle}{\| p\|_{1+\gamma}^\gamma}. \label{eq:psd}
\end{align}
The log-$\gamma$-cross-entropy \cite{gamma-div} is obtained by transforming the PSCE \eqref{eq:psce} using $d_\gamma^{\rm \log}(q,p) = -(1+\gamma)/\gamma \,\log\bigl(-\gamma d_\gamma^{\rm PS}(q,p) +1\bigr)$:
\begin{align}
    d_\gamma^{\log}(q,p) = - \frac{1+\gamma}{\gamma}\,\log\bigl\langle qp^\gamma \bigr\rangle + \log \bigl\langle p^{1+\gamma} \bigr\rangle. \label{eq:log_ce}
\end{align}
In the sense of Definition \ref{def:eqv}, \eqref{eq:psce} and \eqref{eq:log_ce} are equivalent cross-entropies.
However, in this paper, we call \eqref{eq:psce} the PSCE and \eqref{eq:log_ce} the log-$\gamma$-cross-entropy to distinguish between them.
The same naming convention applies to the corresponding divergences.
When we simply say $\gamma$-cross-entropy or divergence, we do not distinguish between these.

\subsection{Density-power-based divergences with functional flexibility}
The FDPD \cite{FDPD} and HD \cite{holder-div1} are known examples of density-power-based divergences with functional flexibility.
Both of these divergences include the DPD and $\gamma$-divergence as special cases, as shown Tables \ref{tab:fdpd} and \ref{tab:hd}.
The definitions of FDPD and HD, as well as their special cases, are detailed in Appendices \ref{sec:fdpd} and \ref{sec:hd}.
Although both the FDPD and HD belong to the class of non-kernel divergences \cite{Jana2019}, their corresponding estimation problems are not necessarily reducible to M-estimation. 
However, except for Jones et al.'s \cite{type0-div} general divergence, which is not used in practical estimation problems, existing divergences can be reduced to M-estimation (Tables \ref{tab:fdpd} and \ref{tab:hd}). 
This provides motivation to investigate density-power-based divergences that can be reduced to M-estimation. 
In the next section, we construct density-power-based divergences with functional flexibility within the class of Bregman divergence that are reducible to M-estimation.

\begin{table}[t]
    \centering
    \caption{Special cases of functional density power divergence \cite{FDPD}}
    \begin{tabular}{c|c}
       Divergence  & M-estimation\\
       \hline\hline
       Density Power \cite{beta-div} & \checkmark\\
       log-$\gamma$ \cite{gamma-div, type0-div}& \checkmark\\
       Bridge Density Power \cite{bridge-div, Gayen2024}& \checkmark\\
       Jones et al.'s general \cite{type0-div}&  $\times$\\
       \hline
    \end{tabular}
    \label{tab:fdpd}
\end{table}
\begin{table}[t]
    \centering
    \caption{Special cases of H\"{o}lder divergence \cite{holder-div1}}
    \begin{tabular}{c|c}
       Divergence & M-estimation\\
       \hline\hline
       Density Power \cite{beta-div}& \checkmark\\
       Transformed Pseudo-Spherical$^1$ \cite{type0-div, gamma-div} & \checkmark\\
       Bregman--H\"{o}lder \cite{holder-div1} & \checkmark\\
       \hline
       \multicolumn{2}{l}{$^{\mathrm{1}}$: The cross-entropy is given by the power transformation of \eqref{eq:psce}.}
    \end{tabular}
    \label{tab:hd}
\end{table}
\renewcommand{\arraystretch}{1.5}
\begin{table*}[h]
    \centering
    \caption{Existing divergences and their unified forms corresponding to $\phi_\gamma$ in NB-DPD with $\gamma>0$}
    \small
    \begin{tabular}{c|c|c|c|c}
         Divergence & CE & $\phi_\gamma(z)$ & Parameters & Special cases \\
         \hline\hline
         Density Power (DP)  \cite{beta-div}& \eqref{eq:dpce} &$z^{1+\gamma}$ & - & -\\
         Pseudo-Spherical (PS)  \cite{gamma-div, type0-div}& \eqref{eq:psce} & $z$ & - & -\\\hline
         \multirow{2}{*}{Bregman--H\"{o}lder (BH) \cite{holder-div1}}& \multirow{2}{*}{\eqref{eq:bhce}} &\multirow{2}{*}{$z^\kappa$} & \multirow{2}{*}{$\kappa\geq1$} &$\kappa=1+\gamma$: DP\\
         &&&&$\kappa=1$: PS \\\hline
         PS-type Bridge Density Power & \multirow{2}{*}{\eqref{eq:psbdpce}} &\multirow{2}{*}{$\frac{1}{\lambda_2}\left[\left(\lambda_1+\lambda_2 z^{1+\gamma}\right)^{\frac{1}{1+\gamma}}-\lambda_1^{\frac{1}{1+\gamma}}\right]$}&\multirow{2}{*}{$\lambda_1\geq0, \lambda_2\geq0$} &$\lambda_1=0, \lambda_2\neq0$: scaled PS\\
         (PSBDP)  \cite{bridge-div, Gayen2024}&&&&$\lambda_1\neq0,\lambda_2\to0$: scaled DP\\\hline
         \multirow{5}{*}{Combination of BH and PSBDP} & \multirow{5}{*}{\eqref{eq:comb_cross_ent}} & \multirow{5}{*}{$\frac{1}{\lambda_2}\left[\left(\lambda_1+\lambda_2 z^{1+\gamma}\right)^{\frac{\kappa}{1+\gamma}}-\lambda_1^{\frac{\kappa}{1+\gamma}}\right]$}&\multirow{5}{*}{$\lambda_1\geq0, \lambda_2\geq0,$}&$\lambda_1=0, \lambda_2\neq0$: scaled BH\\
         &&&&$\lambda_1\neq0, \lambda_2\to0$: scaled DP\\
         &&&&$\lambda_2\neq0, \kappa=1+\gamma$: DP\\
         &&&$\kappa\geq1$&$\lambda_1=0,\lambda_2\neq0, \kappa=1$: scaled PS\\
         &&&&$\lambda_1\neq0,\lambda_2\neq0,\kappa=1$: PSBDP \\\hline
         \multirow{2}{*}{Mixture of DP and PS} & \multirow{2}{*}{\eqref{eq:mix_cros_ent}} & \multirow{2}{*}{$tz^{1+\gamma}+(1-t)z$} & \multirow{2}{*}{$0\leq t\leq1$}&$t=1$: DP\\
         &&&&$t=0$: PS\\\hline
        \multicolumn{4}{l}{Symbol '$\to$' represents limit operation. Abbreviation: CE, cross-entropy.}
    \end{tabular}
    \label{tab:NBDPD}
    \vspace{-10.5pt}
\end{table*}
\renewcommand{\arraystretch}{1}

\section{Norm-based Bregman density power divergence}
Let us assume $\phi_\gamma:\mathbb{R}_+\to\mathbb{R}$ is a strictly increasing and (not strictly) convex function.
Then, for $\gamma>0$, $\phi_\gamma\bigl(\|q\|_{1+\gamma}\bigr)$ becomes a strictly convex functional with respect to the probability distribution $q$.
For any $t\in(0,1)$ and $p,q\in\mathcal{P}$ with $p\neq q$, the following inequality holds:
\begin{align*}
    &\phi_\gamma\bigl(\|tq+(1-t)p\|_{1+\gamma}\bigr)\overset{(a)}{<}\phi_\gamma\bigl(t\,\|q\|_{1+\gamma}+(1-t)\,\|p\|_{1+\gamma}\bigr)\\
    &\overset{(b)}{\leq}t\,\phi_\gamma\bigl(\|q\|_{1+\gamma}\bigr)+(1-t)\,\phi_\gamma\bigl(\|p\|_{1+\gamma}\bigr),
\end{align*}
where (a) follows from the strict convexity of the $L_{1+\gamma}$-norm with respect to $q$ and the monotonic increasing property of $\phi_\gamma$, while (b) follows from the (non-strict) convexity of $\phi_\gamma$.
Hence, we define the strictly convex functional $\Phi_\gamma(q)$ as follows:
\begin{align}
    \Phi_\gamma(q)=
    \begin{cases}
        \frac{\phi_\gamma\bigl(\|q\|_{1+\gamma}\bigr)-\phi_0(1)}{\gamma}, &(\gamma>0),\\
        \phi_{0}'(1)\bigl\langle q\log q\bigr\rangle + \left.\frac{\partial\phi_\gamma(1)}{\partial \gamma}\right|_{\gamma=0}, &(\gamma=0).
    \end{cases}\label{eq:ndbd_ent}
\end{align}
For $\gamma=0$, it is defined by taking the continuous limit.
Using the strictly convex functional \eqref{eq:ndbd_ent}, the corresponding Bregman cross-entropy \eqref{eq:bregman_cross_ent} and divergence \eqref{eq:bregman_div} are defined as follows.
\begin{definition}[Norm-based Bregman density power cross-entropy and divergence]
Let us assume $\phi_\gamma:\mathbb{R}_+\to\mathbb{R}$ is a strictly increasing and (not strictly) convex function, and $\gamma\geq0$. 
When $\gamma>0$, the norm-based Bregman density power cross-entropy (NB-DPCE) between two probability distributions is defined by
\begin{align}
    \begin{split}
            d_{\phi,\gamma}^{\rm NB}(q,p) = 
            &-\frac{\phi_\gamma\bigl(\|p\|_{1+\gamma}\bigr) -\phi_0(1)}{\gamma} \\
            &- \frac{\phi_\gamma'\bigl(\|p\|_{1+\gamma}\bigr)\bigl(\bigl\langle qp^\gamma\bigr\rangle - \bigl\langle p^{1+\gamma}\bigr\rangle\bigr)}{\gamma\,\|p\|_{1+\gamma}^\gamma}, \label{eq:new_ce}
    \end{split}
\end{align}
and when $\gamma=0$, it is defined by
\begin{align}
    d_{\phi,\gamma}^{\rm NB}(q,p) =-\phi_0'(1)\bigl\langle q\log p\bigr\rangle 
        -\left.\frac{\partial\phi_\gamma(1)}{\partial \gamma}\right|_{\gamma=0}, \label{eq:ce0}
\end{align}
where $\phi'_\gamma(z)$ is the derivative of $\phi_\gamma(z)$ with respect to $z$.
When $\gamma>0$, the NB-DPD is defined by
\begin{align}
    \begin{split}
         D_{\phi,\gamma}^{\rm NB}(q,p)=& \frac{\phi_\gamma\bigl(\|q\|_{1+\gamma}\bigr) - \phi_\gamma\bigl(\|p\|_{1+\gamma}\bigr)}{\gamma} \\
         &- \frac{\phi_\gamma'\bigl(\|p\|_{1+\gamma}\bigr)\bigl(\bigl\langle qp^\gamma\bigr\rangle-\bigl\langle p^{1+\gamma}\bigr\rangle\bigr)}{\gamma\,\|p\|_{1+\gamma}^\gamma}, \label{eq:new_div}
    \end{split}
\end{align}
and when $\gamma=0$, it is defined by
\begin{align}
    D_{\phi,\gamma}^{\rm NB}(q,p) = \phi_0'(1)\,\bigl\langle q\log\tfrac{q}{p}\bigr\rangle. \label{eq:kl0}
\end{align}
\end{definition}
We prove that when $\gamma=0$, the NB-DPCE \eqref{eq:ce0} and NB-DPD \eqref{eq:kl0} are defined as continuous limits.
\begin{assumption}\label{assumption:NB-DPD}
    \begin{enumerate}
        \item $\phi_\gamma(z)$ is differentiable with respect to $z$ and $\gamma$.
        \item $\displaystyle \lim_{\gamma\to0} \phi_\gamma\bigl(\|q\|_{1+\gamma}\bigr) = \phi_0\left(1\right)$.
        \item $\displaystyle \lim_{\gamma\to0} \phi'_\gamma\bigl(\|q\|_{1+\gamma}\bigr) = \phi'_0\left(1\right)$.
        \item $\displaystyle \left.\lim_{\gamma\to0}\frac{\partial \phi_\gamma(z)}{\partial \gamma}\right|_{z=\|q\|_{1+\gamma}}
            = \left.\frac{\partial \phi_\gamma\left(1\right)}{\partial \gamma}\right|_{\gamma=0}$.
        \item Suppose that $\phi_\gamma(z)$ has a parameter $\alpha$. If $\alpha$ depends on $\gamma$, then $\lim_{\gamma\to0}\alpha_\gamma = \alpha_0$ exists and it does not depends on $\gamma$. \label{cond:nb_dpd_param}
    \end{enumerate}
\end{assumption}
\begin{theorem}
Under Assumption \ref{assumption:NB-DPD}, when $\gamma\to0$, the NB-DPCE \eqref{eq:new_ce} converges to an affine transformation of the Shannon-cross-entropy \eqref{eq:ce0}, and the NB-DPD \eqref{eq:new_div} converges to a constant multiple of the KL-divergence \eqref{eq:kl0}.
\end{theorem}
\begin{proof}
We calculate the limit of the first term of the NB-DPCE \eqref{eq:new_ce}.
From l’H\^{o}pital’s rule, we have
\begin{align*}
    &\lim_{\gamma\to0}-\frac{\phi_\gamma\bigl(\|p\|_{1+\gamma}\bigr) -\phi_0(1)}{\gamma}
    = -\lim_{\gamma\to0}\frac{\partial \phi_\gamma\left(\|p\|_{1+\gamma}\right)}{\partial \gamma}\\
    =& -\lim_{\gamma\to0}\frac{\partial \phi_\gamma\left(\|p\|_{1+\gamma}\right)}{\partial \|p\|_{1+\gamma}}\frac{\partial \|p\|_{1+\gamma}}{\partial \gamma} - \lim_{\gamma\to0}\left.\frac{\partial \phi_\gamma(z)}{\partial \gamma}\right|_{z=\|p\|_{1+\gamma}}\\
    =&-\phi'_0\left(1\right)\langle p\log p\rangle - \left.\frac{\partial\phi_\gamma(1)}{\partial \gamma}\right|_{\gamma=0}.
\end{align*}
Next, we calculate the limit of the second term of NB-DPCE \eqref{eq:new_ce}.
We have
\begin{align*}
    \lim_{\gamma\to0} -\frac{\phi'_\gamma\left(\|p\|_{1+\gamma}\right)\left(\left\langle qp^\gamma\right\rangle - \left\langle p^{1+\gamma}\right\rangle\right)}{\gamma\|p\|_{1+\gamma}^\gamma}\\
    =-\phi'_0(1)\left(\langle q\log p\rangle - \langle p\log p\rangle\right).
\end{align*}
Thus, we have
\begin{align*}
    \lim_{\gamma\to0} d_{\gamma, \phi}^{\rm NB}(q,p)=-\phi_0'(1)\bigl\langle q\log p\bigr\rangle 
        -\left.\frac{\partial\phi_\gamma(1)}{\partial \gamma}\right|_{\gamma=0}.
\end{align*}
From Definition \ref{def:bregman}, the NB-DPD \eqref{eq:new_div} is defined as the difference between cross-entropies, i.e., $D_{\phi, \gamma}^{\rm NB}(q,p)=d_{\phi, \gamma}^{\rm NB}(q,p)-d_{\phi, \gamma}^{\rm NB}(q,q)$.
Thus, when $\gamma\to0$, it holds that
\begin{align*}
    \lim_{\gamma\to0}D_{\phi, \gamma}^{\rm NB}(q,p) &= \lim_{\gamma\to0}d_{\phi, \gamma}^{\rm NP}(q,p) - \lim_{\gamma\to0}d_{\phi, \gamma}^{\rm NP}(q,q)\\
    &=\phi'_0(1)\left\langle q\log\frac{q}{p}\right\rangle.
\end{align*}
\qed
\end{proof}

\subsection{Existing divergences}
Table \ref{tab:NBDPD} summarizes the existing divergences and their unified forms, defined by the strictly increasing and convex function $\phi_\gamma$.
Here, we explain how existing divergences are special cases of the NB-DPD \eqref{eq:new_div}.
For simplicity, we show the correspondence using cross-entropy.
\begin{example}
    When we set $\phi_\gamma(z)=z^{1+\gamma}$, NB-DPCE \eqref{eq:new_ce} reduces to DPCE \eqref{eq:dpce} \cite{beta-div}.
\end{example}
\begin{example}
    When we set $\phi_\gamma(z)=z$, NB-DPCE \eqref{eq:new_ce} reduces to PSCE \eqref{eq:psce} \cite{gamma-div, type0-div}.
\end{example}
\begin{example}
When we set $\phi_\gamma(z)=z^\kappa, (\kappa\geq1)$, NB-DPCE \eqref{eq:new_ce} reduces to the Bregman-H\"{o}lder cross-entropy (BHCE) \cite{holder-div1} as follows:
\begin{align}
    d_{\kappa, \gamma}^{\rm BH}(q,p)=\frac{\kappa}{\gamma} \|p\|_{1+\gamma}^\kappa\left(1-\frac{1}{\kappa}-\frac{\langle qp^\gamma \rangle}{\langle p^{1+\gamma}\rangle}\right)+\frac{1}{\gamma}. \label{eq:bhce}
\end{align}
When $\kappa=1+\gamma$, it reduces to DPCE \eqref{eq:dpce}, and when $\kappa=1$, it reduces to PSCE \eqref{eq:psce}.
In this way, $\kappa$ actually allowed to depend on $\gamma$; thus, so when considering the limit $\gamma\to 0$, we must assume $\lim_{\gamma\to0}\kappa_\gamma=\kappa_0$ (see condition \ref{cond:nb_dpd_param} of Assumption \ref{assumption:NB-DPD}).
\end{example}
\begin{example}
When we set $\phi_\gamma(z)=\frac{1}{\lambda_2}\Bigl[\bigl(\lambda_1+\lambda_2 z^{1+\gamma}\bigr)^{\frac{1}{1+\gamma}}-\lambda_1^{\frac{1}{1+\gamma}}\Bigr], (\bm{\lambda}=(\lambda_1, \lambda_2)\in\mathbb{R}_+^2)$, 
NB-DPCE \eqref{eq:new_ce} reduces to the pseudo-spherical (PS) type bridge density power cross-entropy (BDPCE) as follows:
\footnotesize
\begin{align}
    d_{\bm{\lambda}, \gamma}^{\rm BDP}(q,p) = \frac{-1}{\lambda_2\gamma}
    \biggl[\frac{\lambda_1+\lambda_2 \langle qp^\gamma \rangle}{
        \bigl(\lambda_1+\lambda_2 \langle p^{1+\gamma}\rangle\bigr)^\frac{\gamma}{1+\gamma}}
    -\bigl(\lambda_1+\lambda_2\bigr)^{\frac{1}{1+\gamma}}\biggr]. \label{eq:psbdpce}
\end{align}
\normalsize
When $\lambda_1=0, \lambda_2\neq0$, BDPCE \eqref{eq:psbdpce} reduces to a constant multiple of the PSCE \eqref{eq:psce}, as follows:
\begin{align*}
    d_{\bm{\lambda},\gamma}^{\rm BDP}(q,p)
    =\lambda_2^{-\frac{\gamma}{1+\gamma}}\,d_\gamma^{\rm PS}(q,p).
\end{align*}
When $\lambda_1\neq0$ and $\lambda_2\to0$, BDPCE \eqref{eq:psbdpce} reduces to a constant multiple of the DPCE \eqref{eq:dpce}, as follows:
\begin{align*}
    \lim_{\lambda_2\to0}d_{\bm{\lambda},\gamma}^{\rm BDP}(q,p)
    =\frac{\lambda_1^{-\frac{\gamma}{1+\gamma}}}{1+\gamma}\,d_{\gamma}^{\rm DP}(q,p).
\end{align*}
\end{example}
\begin{remark}
The BDPCE discussed in previous studies \cite{bridge-div, Gayen2024} is the following log-type cross-entropy:
\footnotesize
\begin{align}
    \begin{split}
        \tilde{d}_{\bm{\lambda}, \gamma}^{\rm BDP}(q,p)=&-\frac{1+\gamma}{\lambda_2\gamma}\log\Bigl(\lambda_1+\lambda_2\langle qp^{\gamma}\rangle\Bigr)\\
        &+ \frac{1}{\lambda_2}\log \Bigl(\lambda_1+\lambda_2\langle p^{1+\gamma}\rangle\Bigr) -\frac{1}{\lambda_2\gamma}\log(\lambda_1+\lambda_2). \label{eq:log_bdpce}
    \end{split}
\end{align}
\normalsize
The relationship between PS-type BDPCE \eqref{eq:psbdpce} and log-type BDPCE \eqref{eq:log_bdpce} is given by $\tilde{d}_{\bm{\lambda}, \gamma}^{\rm BDP}(q,p)=\xi\left(d_{\bm{\lambda}, \gamma}^{\rm BDP}(q,p)\right)$, where
\footnotesize
\begin{align*}
    \xi(z)=-\frac{1+\gamma}{\lambda_2\gamma}\log\left(-\lambda_2\gamma z + \left(\lambda_1+\lambda_2\right)^{\frac{1}{1+\gamma}}\right)-\frac{1}{\lambda_2\gamma}\log\left(\lambda_1+\lambda_2\right).
\end{align*}
\normalsize
Previous studies \cite{bridge-div, Gayen2024} have not shown that the BDPD is a special case of the Bregman divergence in the sense of Definition \ref{def:eqv}.
Previous studies \cite{bridge-div}\footnote{In \cite{bridge-div}, $\bm{\lambda}$ is constrained by $\lambda_1+\lambda_2=1$.}, \cite{Gayen2024} have also pointed out that log-type BDPD can be reduced to log-$\gamma$-divergence \cite{gamma-div} and DPD \cite{beta-div}, but only when ($\lambda_1=0, \lambda_2=1$) and ($\lambda_1=1, \lambda_2\to0$).
In practice, even in the case of the log-type BDPD, it reduces to the log-$\gamma$-divergence \cite{gamma-div} when ($\lambda_1=0, \lambda_2\neq0$) and to the DPD \cite{beta-div} when ($\lambda_1\neq0, \lambda_2\to0$) (see Appendix \ref{sec:fdpd}, Example \ref{example:bdpd}). 
\end{remark}

\subsection{Unified divergences}\label{sec:new_div}
BHD \cite{holder-div1} and BDPD \cite{bridge-div, Gayen2024} are generalizations of the DPD \cite{beta-div} and $\gamma$-divergence \cite{gamma-div, type0-div} respectively, but their relationship has not been clearly elucidated. 
In the previous subsection, we showed that these divergences can be characterized as special cases of the NB-DPD \eqref{eq:new_div}. 
Furthermore, because the NB-DPD is characterized by a strictly increasing and convex function $\phi_\gamma$, it is possible to construct a general divergence based on combinations of $\phi_\gamma$ corresponding to these divergences.
By setting the strictly increasing convex function $\phi_\gamma$ to $\phi_\gamma(z)=\frac{1}{\lambda_2}\Bigl[\bigl(\lambda_1+\lambda_2\,z^{1+\gamma}\bigr)^{\frac{\kappa}{1+\gamma}} - \lambda_1^{\frac{\kappa}{1+\gamma}}\Bigr], (\bm{\lambda}\in\mathbb{R}_+^2, \kappa\geq1)$, we obtain the general form that includes both the BHD and BDPD.
Then, the corresponding cross-entropy is given by
\footnotesize
\begin{align}
    \begin{split}
        d_{\bm{\lambda}, \kappa, \gamma}(q,p) 
        &= \frac{\kappa}{\lambda_2 \gamma}\Bigl(\lambda_1+\lambda_2\left\langle p^{1+\gamma}\right\rangle\Bigr)^{\frac{\kappa}{1+\gamma}}\\
        &\cdot\left[1-\frac{1}{\kappa} 
            -\frac{\lambda_1+\lambda_2\bigl\langle qp^\gamma\bigr\rangle}{\lambda_1+\lambda_2\left\langle p^{1+\gamma}\right\rangle}\right]
            +\frac{1}{\lambda_2\gamma}\left(\lambda_1+\lambda_2\right)^{\frac{\kappa}{1+\gamma}}. \quad \label{eq:comb_cross_ent}
    \end{split}
\end{align}
\normalsize
When $\lambda_1=0, \lambda_2\neq0$, this cross-entropy reduces to $\lambda_2^{\frac{\kappa}{1+\gamma}-1}$ times BHCE \eqref{eq:bhce}.
When $\lambda_1\neq0, \lambda_2\to0$, it reduces to $\kappa\lambda_1^{\frac{\kappa}{1+\gamma}-1}/(1+\gamma)$ times DPCE \eqref{eq:dpce}.
When $\lambda_1\neq0, \lambda_2\neq0, \kappa=1$, it reduces to PS-type BDPCE \eqref{eq:psbdpce}.
For details, see Appendix \ref{sec:example_nb-dpd}.

Moreover, because a convex combination of strictly increasing convex functions is itself a strictly increasing convex function, one can generate general divergences by taking such convex combinations.
For instance, by taking the convex combination of the $\phi_\gamma$ corresponding to the DPD \eqref{eq:dp} and the PSD \eqref{eq:psd}, namely, $\phi_\gamma(z)=tz^{1+\gamma} + (1-t)z, \;(0\leq t\leq 1)$, we obtain a cross-entropy that connects the two:
\begin{align}
    d_{t, \gamma}(q,p)=t\,d_{\gamma}^{\rm DP}(q,p) + (1-t)\,d_{\gamma}^{\rm PS}(q,p). \label{eq:mix_cros_ent}
\end{align}
This essentially yields a convex combination of the DPD \eqref{eq:dp} and PSD \eqref{eq:psd}.
\section{Relationship with other divergences}
In the previous section, we showed that one can generate a generalized divergence by taking the convex combination of $\phi_\gamma$ corresponding to the DPD \eqref{eq:dp} and PSD \eqref{eq:psd} within the NB-DPD framework.
Similarly, one can also generate divergences from the HD \cite{holder-div1} and FDPD \cite{FDPD}, which have functional degrees of freedom, by taking the convex combination of their corresponding functions.
In the NB-DPD \eqref{eq:new_div}, HD \cite{holder-div1}, and FDPD \cite{FDPD} alike, the $\gamma$-divergence is included as a special case.
However, the generated $\gamma$-divergence is only the same in the sense that the cross-entropy is the same according to Definition \ref{def:eqv}; the actual formulas differ.
Therefore, when one generates a divergence by taking the convex combination of the functions corresponding to the DPD and $\gamma$-divergence in each of these divergences, the resulting divergence is different in the sense of Definition \ref{def:eqv}.
The divergences derived from the FDPD \cite{FDPD} and HD \cite{holder-div1} cannot be reduced to M-estimation \eqref{eq:m_estim_ee}, whereas divergences generated from the NB-DPD \eqref{eq:new_div} can be reduced to M-estimation \eqref{eq:m_estim_ee}.
This is a key distinction.

Kanamori and Fujisawa \cite[Th. 3.2]{holder-div1} have shown that the intersection of the Bregman divergence and HD is the BHD.
Because the NB-DPD is a special case of the Bregman divergence, the following corollary holds:
\begin{corollary}
    The intersection between the NB-DPD and HD is limited to the BHD.
\end{corollary}

In addition, the following statement can be made about the NB-DPD and FDPD.
\begin{theorem}\label{thm:NB_FDPD_intersec}
    The intersection between the NB-DPD and FDPD is limited to the BDPD.
\end{theorem}
The proof of Theorem \ref{thm:NB_FDPD_intersec} is provided in Appendix \ref{sec:intersec_fdpd_nb-dpd}.
\section{Robustness}
In robust estimation problems, we assume that the given data are generated from the probability distribution: $(1-\varepsilon)p_\theta(x)+\varepsilon c(x), \;(0\leq\varepsilon<1)$.
Here, $p_\theta(x)$ is the assumed statistical model, and $c(x)$ is the contamination distribution. 
In this case, the objective is to estimate the model parameter $\theta$ while avoiding the adverse effects of outliers generated from the contamination distribution $c(x)$.
As a measure of robustness, the influence function defined by the following expression is well known \cite{robust, huberbook}:
\begin{align*}
    {\rm IF}(x_o, \hat{\theta}, p^*) =\lim_{\varepsilon\to0}\frac{\hat{\theta}\bigl((1-\varepsilon)p^*+\varepsilon c(x_o)\bigr) - \hat{\theta}\bigl(p^*\bigr)}{\varepsilon},
\end{align*}
where $p^*$ is the true distribution.
The influence function represents a perturbation caused by the outlier $x_o$. 
If it diverges to infinity, the estimator collapses. 
In other words, if the influence function is finite, the estimator can be considered robust.
Generally, the influence function of an M-estimator is proportional to $\psi$ in estimating equation \eqref{eq:m_estim_ee}.
For $\gamma>0$, $\psi$ of NB-DPD \eqref{eq:new_div} is bounded, which implies robustness when $\gamma>0$. 
The redescending property is a desirable feature that enables the influence of large outliers to be ignored and is defined through the influence function.
By examining the $\psi$ corresponding to the NB-DPD, it is possible to identify divergences that satisfy the redescending property.
\begin{theorem}\label{thm:redscending}
    Suppose that $\forall\theta\in\Theta, \lim_{x_o\to\infty}p_\theta(x_o)^{1+\gamma}=0$ and $\lim_{x_o\to\infty}p_\theta(x_o)^{1+\gamma}s_\theta(x_o)=0$, where $s_\theta(x)=\partial\log p_\theta(x)/\partial\theta$.
    Then, for the NB-DPD, the redescending property $\lim_{x_o\to\infty} {\rm IF}(x_o, \hat{\theta}, p^*)=0$ holds only for the $\gamma$-divergence.
\end{theorem}
The proof of Theorem \ref{thm:redscending} is provided in Appendix \ref{sec:redscending}.
\section{Conclusion and future work}
In this paper, we constructed a strictly convex functional from the $L_{1+\gamma}$-norm and a strictly increasing convex function $\phi_\gamma$, and defined a Bregman divergence.
We showed that by specifying $\phi_\gamma$, this divergence reduces to the DPD and PSD, as well as their generalizations, namely, the BHD and BDPD. 
Furthermore, by combining the functions $\phi_\gamma$ corresponding to these divergences, one can construct more general divergences. 
We compared the NB-DPD established in this paper with other divergences that have functional flexibility and characterized their properties. 
Finally, we assessed the influence function, a measure of robustness, and showed that the redescending property holds only for the $\gamma$-divergence.

The NB-DPD generates the maximum likelihood estimator when $\gamma=0$, regardless of $\phi_\gamma$, and produces robust estimators with bounded influence function when $\gamma>0$.
Therefore, a future research direction is to explore a $\phi_\gamma$ that offers a better trade-off between robustness and efficiency under the model.
In practical estimation problems, we expect that methods for tuning parameter selection, such as those proposed in \cite{select_tuning, Warwick2005}, could be applied to choose $\phi_\gamma$ appropriately.

\section*{Acknowledgment}
This paper was translated from Japanese into English by OpenAI ChatGPT o1 and 4o.
We would like to thank Editage for English language editing.
This work was supported by JSPS KAKENHI Grant numbers JP23K16849.

\enlargethispage{-1.2cm} 

\bibliographystyle{IEEEtran}
\bibliography{Refj}

\begin{thebibliography}{10}
\providecommand{\url}[1]{#1}
\csname url@samestyle\endcsname
\providecommand{\newblock}{\relax}
\providecommand{\bibinfo}[2]{#2}
\providecommand{\BIBentrySTDinterwordspacing}{\spaceskip=0pt\relax}
\providecommand{\BIBentryALTinterwordstretchfactor}{4}
\providecommand{\BIBentryALTinterwordspacing}{\spaceskip=\fontdimen2\font plus
\BIBentryALTinterwordstretchfactor\fontdimen3\font minus
  \fontdimen4\font\relax}
\providecommand{\BIBforeignlanguage}[2]{{%
\expandafter\ifx\csname l@#1\endcsname\relax
\typeout{** WARNING: IEEEtran.bst: No hyphenation pattern has been}%
\typeout{** loaded for the language `#1'. Using the pattern for}%
\typeout{** the default language instead.}%
\else
\language=\csname l@#1\endcsname
\fi
#2}}
\providecommand{\BIBdecl}{\relax}
\BIBdecl

\bibitem{beta-div}
A.~Basu, I.~R. Harris, N.~L. Hjort, and M.~C. Jones, ``{Robust and efficient
  estimation by minimising a density power divergence},'' \emph{Biometrika},
  vol.~85, no.~3, pp. 549--559, 1998.

\bibitem{Cichocki2010}
A.~Cichocki and S.~Amari, ``Families of alpha- beta- and gamma- divergences:
  Flexible and robust measures of similarities,'' \emph{Entropy}, vol.~12,
  no.~6, pp. 1532--1568, 2010.

\bibitem{scott2001}
D.~W. Scott, ``Parametric statistical modeling by minimum integrated square
  error,'' \emph{Technometrics}, vol.~43, no.~3, pp. 274--285, 2001.

\bibitem{gamma-div}
H.~Fujisawa and S.~Eguchi, ``Robust parameter estimation with a small bias
  against heavy contamination,'' \emph{Journal of Multivariate Analysis},
  vol.~99, no.~9, pp. 2053--2081, 2008.

\bibitem{type0-div}
M.~C. Jones, N.~L. Hjort, I.~R. Harris, and A.~Basu, ``{A comparison of related
  density-based minimum divergence estimators},'' \emph{Biometrika}, vol.~88,
  no.~3, pp. 865--873, 2001.

\bibitem{fujisawa2013}
H.~Fujisawa, ``Normalized estimating equation for robust parameter
  estimation,'' \emph{Electronic Journal of Statistics}, vol.~7, pp.
  1587--1606, 2013.

\bibitem{FDPD}
S.~Ray, S.~Pal, S.~K. Kar, and A.~Basu, ``Characterizing the functional density
  power divergence class,'' \emph{IEEE Transactions on Information Theory},
  vol.~69, no.~2, pp. 1141--1146, 2023.

\bibitem{bridge-div}
A.~K. Kuchibhotla, S.~Mukherjee, and A.~Basu, ``Statistical inference based on
  bridge divergences,'' \emph{Annals of the Institute of Statistical
  Mathematics}, vol.~71, no.~3, pp. 627--656, 2019.

\bibitem{Gayen2024}
A.~Gayen, S.~Roy, and A.~K. Gangopadhyay, ``A unified approach to the
  {P}ythagorean identity and projection theorem for a class of divergences
  based on {M}-estimations,'' \emph{Statistics}, vol.~58, no.~4, pp. 842--880,
  2024.

\bibitem{holder-div1}
T.~Kanamori and H.~Fujisawa, ``Affine invariant divergences associated with
  proper composite scoring rules and their applications,'' \emph{Bernoulli},
  vol.~20, no.~4, pp. 2278--2304, 2014.

\bibitem{Jana2019}
S.~Jana and A.~Basu, ``A characterization of all single-integral, non-kernel
  divergence estimators,'' \emph{IEEE Transactions on Information Theory},
  vol.~65, no.~12, pp. 7976--7984, 2019.

\bibitem{robust}
F.~R. Hampel, E.~M. Ronchetti, P.~J. Rousseeuw, and W.~A. Stahel, \emph{Robust
  Statistics: The Approach Based on Influence Functions}.\hskip 1em plus 0.5em
  minus 0.4em\relax John Wiley \& Sons, 2005.

\bibitem{huberbook}
P.~J. Huber and E.~M. Ronchetti, \emph{Robust Statistics}, 2nd~ed.\hskip 1em
  plus 0.5em minus 0.4em\relax John Wiley \& Sons, 2009.

\bibitem{asympt}
A.~W. van~der Vaart, \emph{Asymptotic Statistics}.\hskip 1em plus 0.5em minus
  0.4em\relax Cambridge University Press, 1998.

\bibitem{Evgeni2018}
E.~Y. Ovcharov, ``Proper scoring rules and {B}regman divergence,''
  \emph{Bernoulli}, vol.~24, no.~1, pp. 53--79, 2018.

\bibitem{Frigyik2008}
B.~A. Frigyik, S.~Srivastava, and M.~R. Gupta, ``Functional {B}regman
  divergence and {B}ayesian estimation of distributions,'' \emph{IEEE
  Transactions on Information Theory}, vol.~54, no.~11, pp. 5130--5139, 2008.

\bibitem{Gneiting2007}
T.~Gneiting and A.~E. Raftery, ``Strictly proper scoring rules, prediction, and
  estimation,'' \emph{Journal of the American Statistical Association}, vol.
  102, no. 477, pp. 359--378, 2007.

\bibitem{bexp-div}
T.~Mukherjee, A.~Mandal, and A.~Basu, ``The {B}-exponential divergence and its
  generalizations with applications to parametric estimation,''
  \emph{Statistical Methods {\&} Applications}, vol.~28, no.~2, pp. 241--257,
  2019.

\bibitem{Roy2019}
S.~Roy, K.~Chakraborty, S.~Bhadra, and A.~Basu, ``Density power downweighting
  and robust inference: Some new strategies,'' \emph{Journal of Mathematics and
  Statistics}, vol.~15, pp. 333--353, 2019.

\bibitem{Singh2021}
P.~Singh, A.~Mandal, and A.~Basu, ``Robust inference using the
  exponential-polynomial divergence,'' \emph{Journal of Statistical Theory and
  Practice}, vol.~15, no.~2, pp. 1--22, 2021.

\bibitem{u-boost}
N.~Murata, T.~Takenouchi, T.~Kanamori, and S.~Eguchi, ``Information geometry of
  ${U}$-boost and {B}regman divergence,'' \emph{Neural Computation}, vol.~16,
  no.~7, pp. 1437--1481, 2004.

\bibitem{select_tuning}
S.~Basak, A.~Basu, and M.~C. Jones, ``On the {‘}optimal{’} density power
  divergence tuning parameter,'' \emph{Journal of Applied Statistics}, vol.~48,
  no.~3, pp. 536--556, 2021.

\bibitem{Warwick2005}
J.~Warwick and M.~C. Jones, ``Choosing a robustness tuning parameter,''
  \emph{Journal of Statistical Computation and Simulation}, vol.~75, no.~7, pp.
  581--588, 2005.

\end{thebibliography}

\onecolumn
\appendices
\section{Functional density power divergence \cite{FDPD}}\label{sec:fdpd}
Let $v:\mathbb{R}_+\to\mathbb{R}\cup\{\pm\infty\}$ be a function such that $w(x)=v(\exp(x))$, where $w:\mathbb{R}\cup\{-\infty\}\to\mathbb{R}\cup\{\pm\infty\}$ is a strictly increasing and convex function.
When $\gamma>0$, the functional density power cross-entroy (FDPCE) and FDPD are defined as
\begin{align}
    d_{v, \gamma}^{\rm FDP}(q,p)&=-\frac{1+\gamma}{\gamma}v\left(\left\langle qp^\gamma \right\rangle\right) + \frac{1}{\gamma}v(1) + 
    v\left(\left\langle p^{1+\gamma}\right\rangle\right), \label{eq:fdpce}\\
    D_{v, \gamma}^{\rm FDP}(q,p)&=\frac{1}{\gamma}v\left(\langle q^{1+\gamma}\right\rangle) - \frac{1+\gamma}{\gamma}v\left(\left\langle qp^\gamma\right\rangle\right) + v\left(\left\langle p^{1+\gamma}\right\rangle\right), \label{eq:fdpd}
\end{align}
respectively.
When $\gamma=0$, the FDPCE and FDPD are defined as
\begin{align*}
    d_{v, \gamma}^{\rm FDP}(q,p)=-v'(1)\left\langle q\log p\right\rangle,\\
    D_{v, \gamma}^{\rm FDP}(q,p)=v'(1)\left\langle q\log\frac{q}{p}\right \rangle,
\end{align*}
respectively.

\begin{example}
    When we set $v(z)=z$, the FDPCE \eqref{eq:fdpce} and FDPD \eqref{eq:fdpd} reduce to DPCE \eqref{eq:dpce} and the DPD \eqref{eq:dp}, respectively.
\end{example}
\begin{example}
    When we set $v(z)=\log z$, the FDPCE \eqref{eq:fdpce} reduces to log-$\gamma$-cross-entropy \eqref{eq:log_ce}, and the FDPD \eqref{eq:fdpd} reduces to the log-$\gamma$-divergence \cite{gamma-div} as follows:
    \begin{align}
        D_{\gamma}^{\log}(q,p)=\frac{1}{\gamma}\log\left(\langle q^{1+\gamma}\right\rangle) - \frac{1+\gamma}{\gamma}\log\left(\left\langle qp^\gamma\right\rangle\right) + \log\left(\left\langle p^{1+\gamma}\right\rangle\right). \label{eq:log_gamma_div}
    \end{align}
\end{example}
\begin{example}
    When we set $v(z)=\ln_{1-\zeta}(z)=(z^\zeta-1)/\zeta, \;(0\leq \zeta\leq1)$, the FDPCE \eqref{eq:fdpce} and FDPD \eqref{eq:fdpd} reduce to Jones et al.'s \cite{type0-div} general cross-entropy and divergence as follows:
    \begin{align*}
        d_{\zeta, \gamma}(q,p)=&-\frac{1+\gamma}{\gamma}\ln_{1-\zeta}\left(\left\langle qp^\gamma \right\rangle\right) + \frac{1}{\gamma}\ln_{1-\zeta}(1) + \ln_{1-\zeta}\left(\left\langle p^{1+\gamma}\right\rangle\right),\\
        D_{\zeta, \gamma}(q,p)=&\frac{1}{\gamma}\ln_{1-\zeta}\left(\langle q^{1+\gamma}\right\rangle) - \frac{1+\gamma}{\gamma}\ln_{1-\zeta}\left(\left\langle qp^\gamma\right\rangle\right)+ \ln_{1-\zeta}\left(\left\langle p^{1+\gamma}\right\rangle\right),
    \end{align*}
    respectively.
    When $\zeta=0$, these cross-entropy and divergence reduce to log-$\gamma$-cross-entropy \eqref{eq:log_ce} and log-$\gamma$-divergence \eqref{eq:log_gamma_div}, respectively.
    From $\lim_{\zeta\to1}\ln_{1-\zeta}z=\log z$, when $\zeta\to1$, these cross-entropy and divergence reduce to the DPCE \eqref{eq:dpce} and DPD \eqref{eq:dp}, respectively.
\end{example}
\begin{example}\label{example:bdpd}
    When we set $v(z)=\log(\lambda_1+\lambda_2 z)/\lambda_2, \;(\bm{\lambda}\in\mathbb{R}_+^2)$, the FDPCE \eqref{eq:fdpce} reduces to the log-type BDPCE \eqref{eq:log_bdpce} and the FDPD \eqref{eq:fdpd} reduce to the log-type BDPD \cite{bridge-div, Gayen2024} as follows:
    \begin{align*}
        \tilde{D}_{\bm{\lambda}, \gamma}^{\rm BDP}(q,p)=\frac{1}{\lambda_2\gamma}\log\left(\lambda_1+\lambda_2\langle q^{1+\gamma}\right\rangle) - \frac{1+\gamma}{\lambda_2\gamma}\log\left(\lambda_1+\lambda_2\left\langle qp^\gamma\right\rangle\right) + \frac{1}{\lambda_2}\log\left(\lambda_1+\lambda_2\left\langle p^{1+\gamma}\right\rangle\right).
    \end{align*}
    When $\lambda_1=0, \lambda_2\neq0$, the BDPCE \eqref{eq:log_bdpce} reduces to a constant
multiple of the log-$\gamma$-cross-entropy \eqref{eq:log_ce}, as follows:
    \begin{align*}
        \tilde{d}_{\bm{\lambda}, \gamma}^{\rm BDP}(q,p) = \frac{1}{\lambda_2}d_{\gamma}^{\log}(q,p).
    \end{align*}
    When $\lambda_1\neq0, \lambda_2\to0$, the BDPCE \eqref{eq:log_bdpce} reduces to a constant
multiple of the DPCE \eqref{eq:dpce}, as follows:
    \begin{align*}
        &\lim_{\lambda_2\to0}\tilde{d}_{\bm{\lambda}, \gamma}^{\rm BDP}(q,p) = \lim_{\lambda_2\to0}\left[-\frac{1+\gamma}{\lambda_2\gamma}\log\Bigl(\lambda_1+\lambda_2\langle qp^{\gamma}\rangle\Bigr) + \frac{1}{\lambda_2}\log \Bigl(\lambda_1+\lambda_2\langle p^{1+\gamma}\rangle\Bigr) -\frac{1}{\lambda_2\gamma}\log(\lambda_1+\lambda_2)\right]\\
        &=\lim_{\lambda_2\to0}\left[-\frac{1+\gamma}{\gamma}\frac{\partial\log\left(\lambda_1+\lambda_2\left\langle qp^\gamma\right\rangle\right)}{\partial\lambda_2}+\frac{\partial\log\left(\lambda_1+\lambda_2\left\langle p^{1+\gamma}\right\rangle\right)}{\partial\lambda_2}-\frac{1}{\gamma}\frac{\partial\log(\lambda_1+\lambda_2)}{\partial\lambda_2}\right]\\
        &=\lim_{\lambda_2\to0}\left[-\frac{1+\gamma}{\gamma}\frac{\left\langle qp^\gamma\right\rangle}{\lambda_1+\lambda_2\left\langle qp^\gamma\right\rangle}+\frac{\left\langle p^{1+\gamma}\right\rangle}{\lambda_1+\lambda_2\left\langle p^{1+\gamma}\right\rangle}-\frac{1}{\gamma}\frac{1}{\lambda_1+\lambda_2}\right]
        =\frac{1}{\lambda_1}d_{\gamma}^{\rm DP}(q,p).
    \end{align*}
\end{example}

\section{H\"{o}lder divergence \cite{holder-div1}}\label{sec:hd}
Let $H:\mathbb{R}_+\to\mathbb{R}$ be a function satisfying $H(z)\geq -z^{1+\gamma}$ for any $z\geq0$ and $H(1)=-1$.
When $\gamma>0$, the H\"{o}lder cross-entropy (HCE) and HD are defined as
\begin{align}
    d_{H, \gamma}(q, p) &= \frac{1}{\gamma}H\left(\frac{\left\langle qp^\gamma \right\rangle}{\left\langle p^{1+\gamma}\right\rangle}\right)\left\langle p^{1+\gamma}\right\rangle + \frac{1}{\gamma}, \label{eq:hce}\\
    D_{H, \gamma}(q, p) &= \frac{1}{\gamma}H\left(\frac{\left\langle qp^\gamma \right\rangle}{\left\langle p^{1+\gamma}\right\rangle}\right)\left\langle p^{1+\gamma}\right\rangle +  \frac{1}{\gamma}\left\langle q^{1+\gamma}\right\rangle, \nonumber
\end{align}
respectively.
The H\"{o}lder entropy is the Tsallis entropy, independent of $H$.
Therefore, in the following, we only show HCE.
\begin{example}
    When we set $H(z)=\gamma-(1+\gamma) z$, the HCE \eqref{eq:hce} reduces to the DPCE \eqref{eq:dpce}.
\end{example}
\begin{example}
    When we set $H(z)=-z^{1+\gamma}$, which is the lower limit of $H$, the HCE \eqref{eq:hce} reduces to the transformation of the PSCE \eqref{eq:psce} as follows:
    \begin{align*}
        d_{H, \gamma}(q,p)=-\frac{1}{\gamma}\frac{\left\langle qp^\gamma \right\rangle^{1+\gamma}}{\left\langle p^{1+\gamma}\right\rangle^\gamma}+\frac{1}{\gamma}.\\
    \end{align*}
\end{example}
\begin{example}
    When we set $H(z)=-\kappa^{\frac{1+\gamma}{\kappa}}|z-1+1/\kappa|^{\frac{1+\gamma}{\kappa}}{\rm sign}(z-1+1/\kappa), \;\kappa\geq1$, the HCE \eqref{eq:hce} reduces to the transformation of the BHCE \eqref{eq:bhce} as follows:
    \begin{align*}
        d_{H, \gamma}(q,p) = -\frac{1}{\gamma}\frac{\left|\kappa\left\langle qp^\gamma\right\rangle -(\kappa-1)\left\langle p^{1+\gamma}\right\rangle\right|^{\frac{1+\gamma}{\kappa}}}{\left\langle p^{1+\gamma}\right\rangle^{\frac{1+\gamma}{\kappa}-1}}
        \cdot{\rm sign}\left(\kappa\left\langle qp^\gamma\right\rangle -(\kappa-1)\left\langle p^{1+\gamma}\right\rangle\right) + \frac{1}{\gamma}.
    \end{align*}
\end{example}

\section{Special cases of Norm-based Bregman density power divergence}\label{sec:example_nb-dpd}
We provide the specific cross-entropy of NB-DPCE \eqref{eq:new_ce} discussed in Section \ref{sec:new_div}.
We also restate the cross-entropy \eqref{eq:comb_cross_ent} that connects BHCE \eqref{eq:bhce} and BDPCE \eqref{eq:psbdpce} as follows:
\begin{align}
    d_{\bm{\lambda}, \kappa, \gamma}(q,p) 
    &= \frac{\kappa}{\lambda_2 \gamma}\Bigl(\lambda_1+\lambda_2\bigl\langle p^{1+\gamma}\bigr\rangle\Bigr)^{\frac{\kappa}{1+\gamma}}\left[1-\frac{1}{\kappa}  -\frac{\lambda_1+\lambda_2\bigl\langle qp^\gamma\bigr\rangle}{\lambda_1+\lambda_2\bigl\langle p^{1+\gamma}\bigr\rangle}\right]
    +\frac{1}{\lambda_2\gamma}\Bigl(\lambda_1+\lambda_2\Bigr)^{\frac{\kappa}{1+\gamma}}. \label{eq:app_comb_cross_ent}
\end{align}
Then the corresponding NB-DPD \eqref{eq:new_div} is given by
\begin{align*}
    D_{\bm{\lambda}, \kappa, \gamma}(q,p)= 
    \frac{\kappa}{\lambda_2 \gamma}\Bigl(\lambda_1+\lambda_2\bigl\langle p^{1+\gamma}\bigr\rangle\Bigr)^{\frac{\kappa}{1+\gamma}}\left[1-\frac{1}{\kappa}  -\frac{\lambda_1+\lambda_2\bigl\langle qp^\gamma\bigr\rangle}{\lambda_1+\lambda_2\bigl\langle p^{1+\gamma}\bigr\rangle}\right]
    +\frac{1}{\lambda_2 \gamma}\Bigl(\lambda_1+\lambda_2\bigl\langle q^{1+\gamma}\bigr\rangle\Bigr)^{\frac{\kappa}{1+\gamma}}.
\end{align*}
When $\lambda_1=0, \lambda_2\neq0, \kappa\geq1$, \eqref{eq:app_comb_cross_ent} reduces to a constant multiple
of the BHCE \eqref{eq:bhce} as follows:
\begin{align*}
    d_{\bm{\lambda}, \kappa, \gamma}(q,p) &= \frac{\kappa}{ \gamma}\lambda_2^{{\frac{\kappa}{1+\gamma}}-1}\bigl\langle p^{1+\gamma}\bigr\rangle^{\frac{\kappa}{1+\gamma}}\left[1-\frac{1}{\kappa}  -\frac{\bigl\langle qp^\gamma\bigr\rangle}{\bigl\langle p^{1+\gamma}\bigr\rangle}\right]
    +\frac{1}{\gamma}\lambda_2^{\frac{\kappa}{1+\gamma}-1}\\
    &=\lambda_2^{\frac{\kappa}{1+\gamma}-1}d_{\kappa, \gamma}^{\rm BH}(q,p).
\end{align*}
When $\lambda_1\neq0, \lambda_2\to0$, \eqref{eq:app_comb_cross_ent} reduces to a constant multiple
of  the DPCE \eqref{eq:dpce} as follows:
\begin{align*}
    &\lim_{\lambda_2\to0}d_{\bm{\lambda}, \kappa, \gamma}(q,p)\\
     =&\lim_{\lambda_2\to0}\left[\frac{\kappa}{\gamma}\frac{\partial}{\partial\lambda_2}\Bigl(\lambda_1+\lambda_2\bigl\langle p^{1+\gamma}\bigr\rangle\Bigr)^{\frac{\kappa}{1+\gamma}}\left[1-\frac{1}{\kappa}  -\frac{\lambda_1+\lambda_2\bigl\langle qp^\gamma\bigr\rangle}{\lambda_1+\lambda_2\bigl\langle p^{1+\gamma}\bigr\rangle}\right]+\frac{1}{\gamma}\frac{\partial}{\partial\lambda_2}\Bigl(\lambda_1+\lambda_2\Bigr)^{\frac{\kappa}{1+\gamma}}\right]\\
     =&\lim_{\lambda_2\to0}\Biggl[\frac{\kappa}{\gamma}\frac{\kappa}{1+\gamma}\left(\lambda_1+\lambda_2\left\langle p^{1+\gamma}\right\rangle\right)^{\frac{\kappa}{1+\gamma}-1}\left\langle p^{1+\gamma}\right\rangle \left[1-\frac{1}{\kappa}  -\frac{\lambda_1+\lambda_2\bigl\langle qp^\gamma\bigr\rangle}{\lambda_1+\lambda_2\bigl\langle p^{1+\gamma}\bigr\rangle}\right]\\
     &-\frac{\kappa}{\gamma}\left(\lambda_1+\lambda_2 \left\langle p^{1+\gamma}\right\rangle\right)^{\frac{\kappa}{1+\gamma}}\frac{\left\langle qp^\gamma\right\rangle\left(\lambda_1+\lambda_2 \left\langle p^{1+\gamma}\right\rangle\right) - \left(\lambda_1+\lambda_2 \left\langle qp^{\gamma}\right\rangle\right)\left\langle p^{1+\gamma}\right\rangle}{\left(\lambda_1+\lambda_2\left\langle  p^{1+\gamma}\right\rangle\right)^2}
     +\frac{1}{\gamma}\frac{\kappa}{1+\gamma}\left(\lambda_1+\lambda_2\right)^{\frac{\kappa}{1+\gamma}-1}\Biggr]\\
     =&\frac{\kappa\lambda_1^{\frac{\kappa}{1+\gamma}-1}}{1+\gamma}\left[ -\frac{1+\gamma}{\gamma}\left\langle qp^\gamma\right\rangle+\left\langle p^{1+\gamma}\right\rangle +\frac{1}{\gamma}\right] =\frac{\kappa\lambda_1^{\frac{\kappa}{1+\gamma}-1}}{1+\gamma}d_{\gamma}^{\rm DP}(q,p).
\end{align*}
When $\lambda_2\neq0, \kappa=1+\gamma$, \eqref{eq:app_comb_cross_ent} reduces to the DPCE \eqref{eq:dpce} as follows:
\begin{align*}
    d_{\bm{\lambda}, \kappa, \gamma}(q,p)
    &= \frac{1+\gamma}{\lambda_2 \gamma}\Bigl(\lambda_1+\lambda_2\bigl\langle p^{1+\gamma}\bigr\rangle\Bigr)\left[\frac{\gamma}{1+\gamma}  -\frac{\lambda_1+\lambda_2\bigl\langle qp^\gamma\bigr\rangle}{\lambda_1+\lambda_2\bigl\langle p^{1+\gamma}\bigr\rangle}\right]
    +\frac{1}{\lambda_2\gamma}\Bigl(\lambda_1+\lambda_2\Bigr)\\
    &= \frac{\lambda_1}{\lambda_2}+\bigl\langle p^{1+\gamma}\bigr\rangle -\frac{\lambda_1(1+\gamma)}{\lambda_2 \gamma}-\frac{1+\gamma}{\gamma}\left\langle qp^\gamma\right\rangle
    +\frac{\lambda_1}{\lambda_2\gamma}+\frac{1}{\gamma}\\
    &= \bigl\langle p^{1+\gamma}\bigr\rangle -\frac{1+\gamma}{\gamma}\left\langle qp^\gamma\right\rangle +\frac{1}{\gamma}=d_{\gamma}^{\rm DP}(q,p).
\end{align*}
When $\lambda_1\neq0, \lambda_2\neq0, \kappa=1$, \eqref{eq:app_comb_cross_ent} reduces to the PS-type BDPCE \eqref{eq:psbdpce} as follows:
\begin{align*}
    d_{\bm{\lambda}, \kappa, \gamma}(q,p)
    &=-\frac{1}{\lambda_2 \gamma}\frac{\lambda_1+\lambda_2\bigl\langle qp^\gamma\bigr\rangle}{\Bigl(\lambda_1+\lambda_2\bigl\langle p^{1+\gamma}\bigr\rangle\Bigr)^{\frac{\gamma}{1+\gamma}}}
    +\frac{1}{\lambda_2\gamma}\left(\lambda_1+\lambda_2\right)^{\frac{1}{1+\gamma}}\\
    &=d_{\bm{\lambda}, \gamma}^{\rm BDP}(q,p).
\end{align*}
When, $\lambda_1=0, \lambda_2\neq0, \kappa=1$, \eqref{eq:app_comb_cross_ent} reduces to a constant multiple of the PSCE \eqref{eq:psce} as follows:
\begin{align*}
    d_{\bm{\lambda}, \kappa, \gamma}(q,p) = \lambda_2^{-\frac{\gamma}{1+\gamma}}d_{\gamma}^{\rm PS}(q,p).
\end{align*}

\section{Proof of Theorem \ref{thm:NB_FDPD_intersec}}\label{sec:intersec_fdpd_nb-dpd}
We prove that the intersection between NB-DPD \eqref{eq:new_div} and FDPD \cite{FDPD} is limited to the BDPD \cite{bridge-div}. 
Since NB-DPD belongs to the class of Bregman divergences, it does not depend on $\phi$ and can be reduced to M-estimation \cite{Gneiting2007}. 
On the other hand, FDPD is not always reducible to M-estimation.
Therefore, it suffices to show that only BDPD can be reduced to M-estimation among the FDPD class.

From the FDPCE \eqref{eq:fdpce}, when we substitute the empirical distribution $\hat{q}(x)=\frac{1}{n}\sum_{i=1}^n \delta(x-x_i)$ for $q$ and the model $p_\theta(x)$ for $p$, the loss function with respect to $\theta$ is as follows:
\begin{align*}
    -\frac{1+\gamma}{\gamma}\,v\Bigl(\frac{1}{n}\sum_{i=1}^n p_\theta(x_i)^\gamma\Bigr)
    + v\Bigl(\bigl\langle p_\theta^{1+\gamma}\bigr\rangle\Bigr),
\end{align*}
where $\delta$ is the Dirac delta function.
By taking the partial derivative of this loss function with respect to $\theta$ and setting it to zero, we derive the following estimating equation:
\begin{align}
    -v'\!\Bigl(\frac{1}{n}\sum_{i=1}^n p_\theta(x_i)^\gamma\Bigr)\frac{1}{n}\sum_{i=1}^n p_\theta(x_i)^\gamma s_\theta(x_i) +v'\!\Bigl(\bigl\langle p_\theta^{1+\gamma}\bigr\rangle\Bigr)\,\bigl\langle p_\theta^{1+\gamma}s_\theta\bigr\rangle=0,
    \label{eq:fdpd_ee}
\end{align}
where $s_\theta(x) = \partial \log p_\theta(x) / \partial \theta$ denotes the score function. In general, the estimating equation for M-estimation takes the following form:
\begin{align*}
    \frac{1}{n}\sum_{i=1}^n \psi(x_i, \theta)=0.
\end{align*}
Hence, the function $\psi$ that depends on $x_i$ must appear in a linearly additive form.
In \eqref{eq:fdpd_ee}, the factor $v'\left(\frac{1}{n}\sum_{i=1}^n p_\theta(x_i)^\gamma\right)$ in the first term must either be constant or be eliminated from the first term by multiplying both sides of \eqref{eq:fdpd_ee} by some value, so that only the remaining part of the first term is preserved.
If $v'(z) = \text{const.}$, \eqref{eq:fdpd_ee} becomes:
\begin{align*}
    -\frac{1}{n}\sum_{i=1}^n p_\theta(x_i)^\gamma s_\theta(x_i)
    + \bigl\langle p_\theta^{1+\gamma} s_\theta\bigr\rangle = 0,
\end{align*}
which corresponds to the following $\psi$:
\begin{align*}
    \psi(x,\theta)
    = - p_\theta(x)^\gamma s_\theta(x) + \bigl\langle p_\theta^{1+\gamma}s_\theta\bigr\rangle.
\end{align*}
In this case, $v(z)=a\,z+b$ with $a>0$ and $b\in\mathbb{R}$, which corresponds to the DPD \cite{beta-div}.
If $v'(z)=1/(\lambda_1+\lambda_2z)$ for $\lambda_1\ge0$ and $\lambda_2>0$, \eqref{eq:fdpd_ee} becomes
\begin{align*}
    -\frac{\frac{1}{n}\sum_{i=1}^n p_\theta(x_i)^\gamma s_\theta(x_i)}{\lambda_1 + \lambda_2 \frac{1}{n}\sum_{i=1}^n p_\theta(x_i)^\gamma}
    + \frac{\bigl\langle p_\theta^{1+\gamma} s_\theta\bigr\rangle}{
      \lambda_1 + \lambda_2\bigl\langle p_\theta^{1+\gamma}\bigr\rangle}
    = 0,
\end{align*}
\begin{align*}
    -\frac{1}{n}\sum_{i=1}^n p_\theta(x_i)^\gamma s_\theta(x_i)\Bigl(\lambda_1 + \lambda_2\bigl\langle p_\theta^{1+\gamma}\bigr\rangle\Bigr) +\bigl\langle p_\theta^{1+\gamma}s_\theta\bigr\rangle
    \Bigl(\lambda_1 + \lambda_2 \frac{1}{n}\sum_{i=1}^n p_\theta(x_i)^\gamma\Bigr) = 0.
\end{align*}
This corresponds the following $\psi$ function:
\begin{align*}
    \psi(x,\theta)
    = -p_\theta(x)^\gamma s_\theta(x)\Bigl(\lambda_1 + \lambda_2 \bigl\langle p_\theta^{1+\gamma}\bigr\rangle\Bigr) + \bigl\langle p_\theta^{1+\gamma}s_\theta\bigr\rangle\Bigl(\lambda_1 + \lambda_2\,p_\theta(x)^\gamma\Bigr).
\end{align*}
In this case, $v(z)=a\log\bigl(\lambda_1+\lambda_2 z\bigr)+b$ with $a>0$ and $b\in\mathbb{R}$, which corresponds to the BDPD \cite{bridge-div, Gayen2024}.
From the above observations, we conclude that among the FDPD, only the BDPD can be reduced to M-estimation.
Therefore, the intersection between the NB-DPD and FDPD is limited to the BDPD.

\section{Proof of Theorem \ref{thm:redscending}}\label{sec:redscending}
We prove that within the NB-DPD \eqref{eq:new_div}, only the PSD \eqref{eq:psd} satisfies the redescending property.
The influence function of the M-estimator is proportional to the $\psi$ function of the estimating equation \eqref{eq:m_estim_ee}:
\begin{align*}
    {\rm IF}(x_o, \hat{\theta}, p^*)=\lim_{\varepsilon\to0}\frac{\hat{\theta}\bigl((1-\varepsilon)p^*+\varepsilon c(x_o)\bigr) - \hat{\theta}\bigl(p^*\bigr)}{\varepsilon}
    \propto\psi(x_o, \theta).
\end{align*}
The influence of large outliers is ignored when $\lim_{x_o\to\infty}{\rm IF}(x_o, \hat{\theta}, p^*)=0$ holds.
This property is referred to as redescending.
In M-estimation, it is sufficient to examine whether the following equation holds
\begin{align*}
    \lim_{x_o\to\infty}\psi(x_o, \theta)=0.
\end{align*}
By substituting an empirical distribution $\hat{q}$ for $q$ and a model for $p_\theta$ in NB-DPCE \eqref{eq:new_ce}, we obtain the following loss function:
\begin{align*}
    d_{\phi,\gamma}^{\rm NB}(\hat{q},p_\theta) = -\frac{\phi_\gamma\bigl(\|p_\theta\|_{1+\gamma}\bigr) -\phi_0(1)}{\gamma} - \frac{\phi_\gamma'\bigl(\|p_\theta\|_{1+\gamma}\bigr)\bigl(\frac{1}{n}\sum_{i=1}^n p_\theta^\gamma(x_i) - \bigl\langle p_\theta^{1+\gamma}\bigr\rangle\bigr)}{\gamma\,\|p_\theta\|_{1+\gamma}^\gamma}.
\end{align*}
The derivative of this loss function with respect to $\theta$ is given by
\begin{align*}
\frac{\partial d_{\phi, \gamma}^{\rm NB}(\hat{q},p_\theta)}{\partial \theta}=&-\frac{1}{\gamma}\frac{1}{\|p_\theta\|_{1+\gamma}^{1+2\gamma}}\biggl[\left\langle p_\theta^{1+\gamma}s_\theta\right\rangle \biggl(\|p_\theta\|_{1+\gamma}\phi_\gamma''\left(\|p_\theta\|_{1+\gamma}\right)-\gamma\phi_\gamma'\left(\|p_\theta\|_{1+\gamma}\right)\biggr)\left(\frac{1}{n}\sum_{i=1}^np_\theta(x_i)^\gamma  - \left\langle p_\theta^{1+\gamma}\right\rangle\right)\\
&+\gamma \phi_\gamma'\left(\|p_\theta\|_{1+\gamma}\right)\left\langle p_\theta^{1+\gamma}\right\rangle\left(\frac{1}{n}\sum_{i=1}^n p_\theta(x_i)^\gamma s_\theta(x_i) - \left\langle p_\theta^{1+\gamma}s_\theta\right\rangle\right)\biggr].
\end{align*}
Therefore, by comparing \eqref{eq:m_estim_ee} and $\frac{\partial d_{\phi, \gamma}^{\rm NB}(\hat{q},p_\theta)}{\partial \theta}=0$, $\psi$ is given by the following equation:
\begin{align*}
    \psi(x,\theta)=& -\left\langle p_\theta^{1+\gamma}s_\theta\right\rangle \biggl(\|p_\theta\|_{1+\gamma}\phi_\gamma''\left(\|p_\theta\|_{1+\gamma}\right)-\gamma\phi_\gamma'\left(\|p_\theta\|_{1+\gamma}\right)\biggr)\left(p_\theta(x)^\gamma  - \left\langle p_\theta^{1+\gamma}\right\rangle\right)\\
&-\gamma \phi_\gamma'\left(\|p_\theta\|_{1+\gamma}\right)\left\langle p_\theta^{1+\gamma}\right\rangle\left( p_\theta(x)^\gamma s_\theta(x) - \left\langle p_\theta^{1+\gamma}s_\theta\right\rangle\right).
\end{align*}
Under assumptions $\lim_{x\to\infty}p_\theta(x)^{1+\gamma}=0$ and $\lim_{x\to\infty}p_\theta(x)^{1+\gamma}s_\theta(x)=0$, 
\begin{align*}
    \lim_{x_o\to\infty} \psi(x_o,\theta) = \left\langle p_\theta^{1+\gamma}\right\rangle \left\langle p_\theta^{1+\gamma}s_\theta\right\rangle\left\|p_\theta\right\|_{1+\gamma}\phi''_\gamma\left(\left\|p_\theta\right\|_{1+\gamma}\right)=0
\end{align*}
holds if and only if $\phi_\gamma''(z)=0$.
The redescending property holds only when $\phi_\gamma(z)=az+b, \;a>0, b\in\mathbb{R}$, that is, in the case of the PSD \eqref{eq:psd}.

\end{document}